%% file: compressforward.tex
\documentclass[letterpaper,conference,dvipsnames]{IEEEtran}
\usepackage{tikz}
\usetikzlibrary{decorations.pathreplacing}
\usetikzlibrary{shapes.misc, positioning}
\usepackage{microtype}
\usepackage{relsize}
\usepackage{setspace}
\usepackage{algorithm}
\usepackage{algorithmicx}
\usepackage{algpseudocode}
\usepackage{bbm}
\usepackage{epsfig}  
\usepackage{graphicx}  
\usepackage{longtable}
\usepackage{amsmath}  
\usepackage{array}
\usepackage{multirow}
\usepackage{rotating}
\usepackage{amssymb}  
\usepackage{lscape}  
\usepackage{amsthm}
\usepackage[justification=raggedright]{caption}
\usepackage[margin=0.75in]{geometry}
\allowdisplaybreaks
\newtheorem{theorem}{Theorem}
\newtheorem{corollary}{Corollary}
\newtheorem{lemma}{Lemma}

\def\forigin{0.1}

\def\forighei{1.8}

\def\fchwidthsh{1.9}
\def\foffonesh{2.8}

\def\foffonevs{2.5} 
\def\fchwidthvs{1.5} 

\def\fchwidthcfd{1.1}
\def\foffonecfd{1.5}

\def\fchwidthsfd{1.1}
\def\foffonesfd{1.2}

\def\foffonecma{2.8}
\def\fchwidthcma{1.9}
\def\fdheightcma{1.8}

\def\foffonemrc{2.8}
\def\fchwidthmrc{1.9}
\def\fdheightmrc{1.4}

\def\foffonetsm{2.4}
\def\fchwidthtsm{1.9}
\def\fdheighttsm{1.6}

\def\offlab{0.25}

\def\fdheight{1.2}
\def\orighei{1.8}

\ifCLASSINFOpdf
\else
\fi
\begin{document}
%
\title{Compress-Forward Schemes for General Networks}
%
%
%
\author{\IEEEauthorblockN{Jonathan Ponniah~\IEEEmembership{Member,~IEEE,}
}
\IEEEauthorblockA{Department of Electrical Engineering\\
San Jose State University
}
\thanks{This material is based upon work partially supported by NSF Contract CNS-1302182,  AFOSR Contract FA9550-13-1-0008, and NSF Science \& Technology Center Grant CCF-0939370.}}

%



\maketitle

\begin{abstract}
Compress-forward (CF) schemes are studied in general networks.  The CF rate for the one-relay channel defines outerbounds on both the CF rate for general networks and the compression rate-vector region supporting this rate.  We show the outerbound is achievable using regular decoding with constant encoding delays, avoiding the exponential delays and restrictions on bidirectional communication in noisy network coding and backward decoding.  The concept of \textit{layering} is introduced to harmonize regular CF schemes with the framework of \textit{flow decomposition} in the decode-forward setting.  Layerings correspond to regular decoding schemes.  Any desired compression rate-vector in the outerbound is achievable by some layering, which is found using the same ``shift'' operation in flow decomposition.  In separate work, we show that ``shifting'' minimizes the operations needed to find layerings and thus minimizes the complexity of the compression rate-vector region.
\end{abstract}


%
\IEEEpeerreviewmaketitle

\input{introduction}
\input{layerings}

\input{mainresult}
\section{Proof of Theorem \ref{theoremone}}
\label{prooftheoremone}

To prove Theorem \ref{theoremone}, we first fix ${\bf\hat{R}}\in{\cal\hat{R}}_{d}$ and pick an arbitrary layering ${\bf L}_{d}$.  Then we find the largest subset $S\subseteq{\cal S}(d)$ for which ${\bf\hat{R}}$ violates (\ref{compdecomp3}) and use this subset to ``shift" ${\bf L}_{d}$.  We repeat this process until ${\bf\hat{R}}$ satisfies (\ref{compdecomp3}) for all $S\subseteq{\cal S}(d)$.  Figure \ref{polytope} depicts the geometric relationship between different layerings (i.e., regular decoding schemes) and their achievable regions for the two-relay and three-relay channels..

Fix ${\bf\hat{R}}\in{\cal\hat{R}}$ and let $U$ denote the largest subset of ${\cal S}$ that violates (\ref{compdecomp3}) and define ${\bf L}^{\prime}:=\text{\sc shift}({\bf L},U)$, where for every $i\in{\cal S}$:  
\begin{align}
\label{shift2}
\text{\sc layer}^{\prime}(i)=\begin{cases}l&i\in A_{l}({\cal S})\setminus A_{l}(U)\\ l+1 & i\in A_{l}(U)\end{cases}	
\end{align}


For ${\bf L}^{\prime}_{d}=\text{\sc shift}({\bf L}_{d},U)$, let ${\hat{\cal R}}({\bf L}^{\prime}_{d})$ denote the set of rate vectors that satisfy:

\begin{align}
\nonumber
\hat{R}_{S}&<\displaystyle\sum_{i\in S}H(X_{i}\hat{Y}_{i})\\
\label{compdecomp2}
&\hspace{4mm}-\displaystyle\sum^{|{\bf L}^{\prime}_{d}|}_{l=0}H(X_{A^{\prime}_{l}(S)}\hat{Y}_{A^{\prime}_{l-1}(S)}|X_{\tilde{A}^{\prime}_{l}(S)}\hat{Y}_{\tilde{A}^{\prime}_{l-1}(S)}Y_{d})
\end{align}
for all $S\subseteq{\cal S}(d)$, where (\ref{A}) and (\ref{Atilde}) define $A^{\prime}_{l}(\cdot)$ and $\tilde{A}^{\prime}_{l}(\cdot)$ respectively with respect to ${\bf L}^{\prime}_{d}$.  
 Lemmas \ref{disjoint}-\ref{join} correspond to \cite{flowdecomparxiv}:(Lemmas 3-5).  The proofs here are considerably simpler due to the absence of flows in the compress-forward-setting.  
 
 
\input{disjoint}

\input{remain}

\input{join}

Lemmas \ref{induction} and \ref{grandefinale} correspond to \cite{flowdecomparxiv}:(Lemmas 6 and 7).  
The following definition will be useful.  Let $\{B_{l}\subseteq C_{l}\subseteq{\cal S}\}$ be $l$-indexed sequences of sets in which $\tilde{B}_{l}:=C_{l}\setminus B_{l}$.  Define:
\begin{align}
\label{hdef}
	h[B_{l}|\tilde{B}_{l}]:=H(X_{B_{l}}\hat{Y}_{B_{l-1}}|X_{\tilde{B}_{l}}\hat{Y}_{\tilde{B}_{l-1}}Y_{d}).
\end{align}

Let $Z\subseteq{\cal S}(d)$ denote the set of source nodes that satisfy (\ref{compdecomp3}) for all $S\subseteq Z$, and let $Z^{\prime}\subseteq{\cal S}(d)$ denote the set of source nodes that satisfy (\ref{compdecomp2}) for all $S\subseteq Z^{\prime}$.
  
\input{induction}

\input{lemma7}

\section{Conclusion}
\label{conclusion}
An outerbound on the CF rate and compression rate-vector region derived from the one-relay channel was shown to be achievable in general networks, using regular decoding schemes with constant  encoding delays in the channel usage.  Regular coding avoids the exponential delays and restrictions on bidirectional communication in noisy network coding and backward decoding.  Layerings were introduced, which correspond to regular coding schemes.  The same shift operation in flow decomposition for DF schemes was used to find layerings that achieve any desired compression rate-vector in the outerbound.  This approach harmonizes the proofs for CF and DF schemes, setting the stage for a united CF-DF framework.  In separate work, we show that the shifting approach minimizes the complexity of the compression rate region.
\bibliographystyle{IEEEtran}
\bibliography{compressforward}
\end{document}

%% file: introduction.tex
\section{Introduction}
Alice, Boris, and Charlie work at different stations on the floor of a noisy machine shop.  Charlie can't quite hear the instructions Alice yells from across the floor.  Boris, who doesn't speak English, tries to help by yelling the closest-sounding Russian words to the English words he hears from Alice.  Charlie attempts to decipher Alice's instructions from the combined vocal output of Alice and Boris.  This contrived scenario is a loose description of the original compress-forward (CF) scheme for the three-node relay channel, where a source (Alice) communicates through a noisy channel to a destination (Charlie), with the help of a relay (Boris).  Another scheme called decode-forward (DF) covers the more natural scenario in which a bilingual Boris yells out Alice's instructions in English.  Of course, Boris must be able to accurately decode these instructions in the first place.  Neither scheme is universally better than the other.

Although its capacity is unknown, the relay channel is still a useful building block for studying general networks from an information-theoretic perspective.  Network coding, which is provably optimal in single-source multi-cast networks with noiseless links \cite{LiNetworkCoding2003}, is a special case of CF \cite{noisynetworkcoding}.  The “spirit” of network coding is also present in DF schemes through index-coding \cite{Xie2007}, a related strategy with a long history in networking and communication \cite{indexcoding}.  Issues affecting DF and CF schemes thus have profound implications for general networks and must be understood fully.

The first issue is that DF schemes in bidirectional networks are subject to a fundamental tension in which each relay has an incentive to wait for others to decode first \cite{Ponniah2008}.  This tension completely disables “backward-decoding”, a powerful information-theoretic technique \cite{XieKumar}.  For instance, Noisy Network Coding (NNC) which extends network coding to general channels, belongs to the family of backward-decoding schemes.  Although CF schemes are not directly subject to the tension in DF schemes, both are linked in joint DF-CF schemes which are not universally beaten by either in isolation.  

A second issue is that backward-decoding schemes experience long encoding delays (exponential in the channel usage), which render them functionally impractical.  Another strategy called “regular coding” has short encoding delays (constant in channel usage) and supports bidirectional communication.  However, regular coding generally achieves lower rates than backward decoding, which leads to a third issue; the rate achieved by any proposed family of CF schemes in general networks, must be “optimal” in some sense.  Optimality is a delicate term because the capacity regions of all general multi-terminal channels are unsolved.  Nevertheless, it is desirable to separate networking problems from the ``physical layer'' problems not fully understood in multi-terminal channels.  

A fourth issue is the complexity of any achievable region; the number of operations required to verify whether a desired rate-vector is included in the region (also of importance is the number of operations required to find the scheme that achieves a desired rate-vector).  In the CF setting, complexity arises from deciding whether the compression rate-vector  supporting the actual information rate, is achievable.   

We address these issues in this paper (touching upon complexity but leaving fuller treatment for separate work).  The mutual-information constraints derived from the one-relay CF rate, define an outer-bound on the CF rate for general networks.  We show these constraints are also sufficient with regular coding schemes, avoiding the long encoding delays and restrictions on bidirectional communication in backward-decoding.  This result implies the CF rate has minimal complexity.  

The key idea is \textit{layerings}; ordered partitions of nodes that correspond to regular decoding schemes.  Layerings are also present in the DF setting, accompanied by so-called flows; the sequences of nodes or routes that forward messages from each source \cite{flowdecomparxiv}.  Flows do not appear in the CF setting because nodes do not forward actual messages, only ``Boris-style'' approximations of the ``sounds they hear''.  Hence, the proof in this paper is a greatly simplified version of the proof in the DF setting, and in fact, a useful intermediary with similar lemmas, claims, and sub-claims.  As in \cite{flowdecomparxiv}, we define a shift operation that alters layerings so their achievable regions are closer to a target compression rate-vector.  For any arbitrary target satisfying the CF outer-bound, we construct a sequence of shifted layerings and show the target is eventually included in the achievable region.  

The rest of the paper is organized as follows:  Section \ref{background} provides a survey of some previous work.  Section \ref{OutlineAndPreliminaries} provides a high-level overview of the main results and an outline of the proof.  The concept of layering is introduced in Section \ref{layeredpartitions}, and the main result is presented in Section \ref{mainresult}.  Section \ref{prooftheoremone} includes the proof and Section \ref{conclusion} concludes the paper.

\section{Literature Review}
\label{background}
The relay channel was introduced in \cite{Meulen1971}.  CF and DF schemes for the one-relay channel were proposed in \cite{GamalCover}.  The original DF scheme, which combined super-position coding, random binning, and list-decoding, was simplified and streamlined for multi-relay channels in \cite{KramerGastpar} and \cite{XieMultiLevel}.  The binning and list decoding strategy was replaced by a joint typicality decoding scheme called sliding-window decoding, first used for the “multiple-access channel with generalized feedback” in \cite{Carleial1982}.  Sliding-window decoding belongs to the regular decoding family of schemes used in this paper.  A fundamentally different scheme from sliding window decoding called backward-decoding, was first proposed for the “multiple-access channel with cribbing encoders” \cite{willems}.  

Noisy Network Coding (NNC) \cite{noisynetworkcoding} is a CF scheme that generalizes the network coding scheme in \cite{Avestimehr2011}.  It turns out that backward-decoding generalizes NNC \cite{Wu2013}\cite{HouKramer}.  Backward-decoding achieves higher rates than sliding-window decoding in general multi-source multi-relay channels \cite{XieKumar}, but requires much longer encoding delays.  A variation of regular coding called ``offset encoding'' was proposed to circumvent the delay problem in the multiple-access relay channels (MARC) \cite{Sankar}.  Three different offset encoding schemes collectively achieve the same region as backward-decoding in \cite{Sankar} thus solving the delay problem for the MARC.  

Backward-decoding cannot support bidirectional communication in the DF setting \cite{XieKumar} due to the tension between relays.  The offset-encoding scheme was first applied in the two-way two-relay channel, in an attempt to avoid this tension \cite{Ponniah2008}.  A parallel and independent effort applied offset-encoding to the CF setting \cite{Yassaee}.  Key components of our framework appear in \cite{Yassaee} including the mapping between layerings and regular decoding schemes, the typicality checks and associated error probabilities, and the description of the CF rate and compression rate-region.  However, the proof in \cite{Yassaee} omits details that call into question the viability of the overall approach.  By contrast, the proof we present relies on the ``shift'' operation in \cite{flowdecomparxiv}.

The CF schemes here and in \cite{Yassaee} rely on regular coding, and open the door to joint DF-CF schemes that avoid the long encoding delays and restrictions on bidirectional communication in joint schemes based on backward decoding \cite{Wu2014}\cite{Hou2016}. 

\section{Outline and Preliminaries}
\label{OutlineAndPreliminaries}
We provide a rough overview of the proof and main result, saving the rigor for the sections to follow.  The one-relay channel consists of a source (node 1), a destination (node 3), and a relay (node 2).  The discrete memoryless channel  $({\cal X}_{1}\times{\cal X}_{2}, p(y_{3},y_{2}|x_{1},x_{2}),{\cal Y}_{2}\times{\cal Y}_{3})$ models the channel dynamics.
For $p=p(x_{1})p(x_{2})p(\hat{y}_{2}|x_{2},y_{2})p(y_{3},y_{2}|x_{1},x_{2})$ any rate satisfying, 
\begin{align}
\label{CFrate}
R<\max_{p}I(X_{1};\hat{Y}_{2}Y_{3}|X_{2}) 	
\end{align}
is achievable \cite{GamalCover} provided:
\begin{align}
\label{singlerelay}
	I(\hat{Y}_{2};Y_{2}|X_{2}Y_{3})\leq I(X_{2};Y_{3}).
\end{align}
The constraint in (\ref{singlerelay}) bounds the relay compression rate while the (\ref{CFrate}) bounds the source message rate.  
Now consider a channel with nodes ${\cal N}=\{1,\ldots,|{\cal N}|\}$.  Let ${\bf y}_{\cal N}:=(y_{1},\ldots,y_{|{\cal N}|})$ and ${\bf x}_{\cal N}:=(x_{1},\ldots,x_{|{\cal N}|})$.  The input-output dynamics conform to the discrete memoryless channel:
\begin{align}
	\label{discretememorylesschannel}
	(\displaystyle\prod_{i\in {\cal N}}{\cal X}_{i},\hspace{1mm}p({\bf y}_{\cal N}|{\bf x}_{\cal N}),\displaystyle\prod_{i\in{\cal N}}{\cal Y}_{i}).
\end{align}
The constraints in (\ref{CFrate}) and (\ref{singlerelay}) define the following outer-bound on the CF rate for general networks.  Let ${\cal S}(d):=\{2,\ldots,d-1\}$ and $d=|{\cal N}|$.  For $p:=p(x_{1})[\prod_{i\in{\cal S}(d)}p(x_{i})p(\hat{y}_{i}|x_{i}y_{i})]p({\bf y}_{{\cal N}}|{\bf x}_{{\cal N}})$ the source rate must satisfy:
\begin{align}
	\label{CFouterbound}
	R<\max_{p}I(X_{1};\hat{Y}_{{\cal S}(d)}Y_{d}|X_{{\cal S}(d)}),
\end{align}
where: 
\begin{align}
\nonumber
&I(\hat{Y}_{S};Y_{S}|X_{{\cal S}(d)}\hat{Y}_{{\cal S}(d)\setminus S}Y_{d})\\
\label{compressionouter}
&\hspace{30mm}\leq I(X_{S};\hat{Y}_{{\cal S}(d)\setminus S}Y_{d}|X_{{\cal S}(d)\setminus S}),	
\end{align}
for all $S\subseteq{\cal S}(d)$.  The main result in Theorem \ref{theoremone} is that (\ref{CFouterbound}) and (\ref{compressionouter}) are also sufficient, which also implies they have minimal complexity.

A layering ${\bf L}_{d}$ is an ordered partition of ${\cal S}(d)$ that defines a scheme for decoding relay compressions.  For any compression rate vector ${\bf\hat{R}}$ satisfying (\ref{compressionouter}), we show there is a layering that achieves ${\bf\hat{R}}$.  To prove Theorem \ref{theoremone}, we pick an arbitrary ${\bf\hat{R}}$ consistent with (\ref{compressionouter}) and an arbitrary layering ${\bf L}_{d}$.  If ${\bf L}_{d}$ does not achieve ${\bf\hat{R}}$, we define the following ``shift'' operation:
\begin{align}
{\bf L}^{\prime}_{d}&=\text{\sc shift}({\bf L}_{d},S),
\end{align}
where $S$ is a selected subset of relays whose compressions are decoded by node $d$.  Lemma \ref{induction} shows that the compression rate region achieved by ${\bf L}^{\prime}_{d}$ is closer to ${\bf\hat{R}}$ than ${\bf L}_{d}$.  The proof of Lemma \ref{induction} relies on Lemmas \ref{disjoint}, \ref{remain} and \ref{join}.

Next, we create a sequence of layerings $\{{\bf L}_{d,n}:n\in\mathbb{N}\}$, where ${\bf L}_{d,n+1}=\text{\sc shift}({\bf L}_{d,n}, S_{n})$ and $\{S_{n}:n\in\mathbb{N}\}$ is a selected sequence of subsets of ${\cal S}(d)$.  Lemma \ref{grandefinale} proves there is some $n^{*}\in\mathbb{N}$, such that ${\bf L}_{d,n^{*}}$ achieves ${\bf\hat{R}}$.  The proof of Lemma \ref{grandefinale} uses Lemma \ref{induction}.  We aim to emulate the DF setting in \cite{flowdecomparxiv} as much as possible, including the labels assigned to claims and sub-claims.  Occasionally, certain label indices might be skipped to maintain this correspondence.   

The following definition of typicality is used.  Let $X_{{\cal N}}:=\{X_{i}:i\in{\cal N}\}$ denote a finite collection of discrete random variables  with a fixed joint distribution $p(x_{{\cal N}})$ for some $x_{{\cal N}}:=\{x_{i}\in{\cal X}_{i}:i\in{\cal N}\}$.  Similarly, let ${\bf x}_{i}:=\{x^{(m)}_{i}\in{\cal X}_{i}:1\leq m\leq n\}$ denote an $n$-length vector of ${\cal X}_{i}$ and let ${\bf x}_{{\cal N}}:=\{{\bf x}_{i}:i\in{\cal N}\}$.  The set of typical $n$-sequences is given by: 
\begin{align}
	\nonumber
	&\hspace{-1mm}T^{(n)}_{\epsilon}(X_{{\cal N}}):=\\
	\nonumber
	&\hspace{7mm}\bigg\{{\bf x}_{{\cal N}}:\left|-\frac{1}{n}\log\text{Prob}({\bf x}_{S})-H(X_{S})\right|<\epsilon, \forall S\subseteq{\cal N}\bigg\},
\end{align}
where $\text{Prob}({\bf x}_{S}):=\prod^{n}_{m=1}p(x^{(m)}_{S})$.

%% file: layerings.tex
\section{Layerings}
\label{layeredpartitions}
The relays encode their compressions as per Wyner-Ziv; node $i\in{\cal S}(d)$ is assigned a bin codebook and an indexed family of compression codebooks, each generated from a bin codeword.  The compression codewords in each codebook are evenly distributed into bins.  At the end of each block, node $i\in{\cal S}(d)$ finds a compression codeword jointly typical with its observations, then transmits the corresponding bin index in the next block.  Node $i$ selects the compression codebook that corresponds to the bin index it transmits in the current block.
  
\begin{itemize}
	\item For each node $i\in{\cal S}(d)$, generate a bin codebook of $2^{nR_{i}}$ i.i.d codewords ${\bf x}_{i}(m)$ according to $p(x_{i})$, where $m\in\{1,\ldots,2^{nR_{i}}\}$ denotes the bin index. 
	\item For each bin codeword ${\bf x}_{i}(m)$, generate a compression codebook of $2^{n\hat{R}_{i}}$ i.i.d codewords ${\bf \hat{y}_{i}}(w|m)$ according to $p(\hat{y}_{i}|x_{i})$, where $w\in\{1,\ldots,2^{n\hat{R}_{i}}\}$ denotes the compression index.
	\item Distribute the compression codewords in this codebook evenly into bins indexed by $m_{i}\in\{1,\ldots,2^{nR_{i}}\}$.  Assign the compression codewords in each bin an index $m^{\prime}_{i}\in\{1,\ldots,2^{n(\hat{R}_{i}-R_{i})}\}$ so that $w_{i}:=(m^{\prime}_{i},m_{i})$.
\end{itemize}

Encoding occurs over $B$ blocks of $n$ channel uses.  By assumption, node $i$ knows the bin index $m_{i}(b-1)$ it will transmit at the start of block $b$.
\begin{itemize}
	\item In block $b$, node $i$ sends the bin codeword ${\bf x}_{i}(m_{i}(b-1))$.
	\item At the end of block $b$, node $i$ finds the compression codeword ${\bf\hat{y}_{i}}(w_{i}(b)|m_{i}(b-1))$ jointly typical with its observed sequence ${\bf y_{i}}(b)$ where $w_{i}(b):=(m^{\prime}_{i}(b),m_{i}(b))$.
	\item Node $i$ will find such a compression codeword with high probability if $\hat{R}_{i}>I(\hat{Y}_{i};Y_{i}|X_{i})$.
\end{itemize}

The destination decoding scheme relies on \textit{layerings} which are ordered partitions of the relay nodes.  A layering ${\bf L}_{d}:=(L_{0},L_{1},\ldots,L_{|{\bf L}_{d}|-1})$ of ${\cal S}(d)$ satisfies the following conditions by definition:

{(L1)} $L_{l}\subseteq{\cal S}(d)$ for every $l=0,\ldots,|{\bf L}_{d}|-1$,

{(L2)} $L_{l}\cap L_{q}=\{\}$ for $l\neq q$,

{(L3)} ${\cal S}(d)=\cup^{|{\bf L}_{d}|-1}_{l=0}L_{l}$,

{(L4)} $L_{|{\bf L}_{d}|-1}\neq\{\}$,

{(L5)} $\text{\sc layer}(i)=l$ if $i\in L_{l}$.

The sets in ${\bf L}_{d}$ can be empty provided (L4) is satisfied.  The ``order'' of a layer is reference to time; deeper layers refer to observations deeper in the past.  

In block $b$, the destination decodes the compression vector ${\bf w}(b):=(w_{2},\ldots,w_{d-1})$ where for every $i\in{\cal S}(d)$: 
\begin{align}
\label{wdb}
w_{i}:=w_{i}(b-\text{\sc layer}(i)-1).	
\end{align}
Two rules characterize the decoding scheme.  First, the destination always decodes the bin index of a compression before the compression itself, using the shared correlation between its observations and the compression to identify the compression from the bin.  Second, the compressions of relays in shallow layers help the destination decode the compressions of relays in deeper layers.
  
For any subset $S\subseteq{\cal S}(d)$ and $0\leq l\leq |{\bf L}_{d}|-1$, let:
\begin{align}
	\label{A}
	A_{l}(S)&:=S\cap L_{l},\\
	\label{Atilde}
	\tilde{A}_{l}(S)&:=\{\cup^{l}_{q=0}L_{q}\}\setminus A_{l}(S).
\end{align} 

The set $A_{l}(S)$ is the subset of $S$ ``active'' in layer $l$.  The dependence of $A_{l}(\cdot)$ and $\tilde{A}_{l}(\cdot)$ on a particular ${\bf L}_{d}$ is implied.  In block $b$, the destination node $d$ decodes ${\bf w}(b)$ as defined in (\ref{wdb}) by finding the compression vector ${\bf\hat{w}}(b):=(\hat{w}_{2},\ldots,\hat{w}_{d-1})$ that satisfies the following typicality checks for $0\leq l\leq |{\bf L}_{d}|$:
\begin{align}
\nonumber
	&\hspace{-3.5mm}(\{{\bf x}_{i}(\hat{m}_{i}):i\in A_{l}({\cal N})\},\\
	\nonumber
	&\hspace{1.4mm}\{{\bf\hat{y}}_{i}(\hat{w}_{i}|m_{i}(b-l-2)):i\in A_{l-1}({\cal N})\},\\
	\nonumber
	&\hspace{1.4mm}\{{\bf X}_{i}(b-l):i\in\tilde{A}_{l}({\cal N})\},\\
	\nonumber
	&\hspace{1.4mm}\{{\bf\hat{Y}}_{i}(b-l):i\in\tilde{A}_{l-1}({\cal N})\},{\bf Y}_{d}(b-l)\})\\
	\label{typpie}
	&\hspace{1.4mm}\in T^{(n)}_{\epsilon}(X_{\{A_{l}({\cal N})\cup\tilde{A}_{l}({\cal N})\}},\hat{Y}_{\{A_{l-1}({\cal N})\cup\tilde{A}_{l-1}({\cal N})\}},Y_{d}),
\end{align}
where $\hat{w}_{i}:=(\hat{m}^{\prime}_{i},\hat{m}_{i})$ for all $i\in{\cal S}(d)$.  
An error occurs if some subset $S\subseteq{\cal S}(d)$ of the compression estimates $\{\hat{w}_{i}: i\in S\}$  are incorrect.  For each typicality check $l=0,\ldots,|{\bf L}_{d}|$, the number of jointly typical codewords $(\{{\bf x}_{i}:i\in A_{l}(S)\},\{{\bf\hat{y}}_{i}:i\in A_{l-1}(S)\})$ conditioned on $\{{\bf X}_{i}:i\in\tilde{A}_{l}(S)\}$, $\{\hat{\bf Y}_{i}:i\in\tilde{A}_{l}(S)\}$, and ${\bf Y}_{d}$ is approximately: 
\begin{align}
\nonumber
\exp_{2}(nH(X_{A_{l}(S)}\hat{Y}_{A_{l-1}(S)}|X_{\tilde{A}_{l}(S)}\hat{Y}_{\tilde{A}_{l-1}(S)}Y_{d}))	
\end{align}
The probability that independent codewords in $(\{{\bf x}_{i}:i\in A_{l}(S)\},\{{\bf\hat{y}}_{i}:i\in A_{l-1}(S)\})$ are jointly typical with $\{{\bf X}_{i}:i\in\tilde{A}_{l}(S)\}$, $\{\hat{\bf Y}_{i}:i\in\tilde{A}_{l}(S)\}$, and ${\bf Y}_{d}$ is approximately:

\vspace{3mm}
$\hspace{9mm}\exp_{2}(n(H(X_{A_{l}(S)}\hat{Y}_{A_{l-1}(S)}|X_{\tilde{A}_{l}(S)}\hat{Y}_{\tilde{A}_{l-1}(S)}Y_{d}))$

$\hspace{13mm}\times\exp_{2}(-n(\sum_{i\in A_{l}(S)}H(X_{i})))$ 

$\hspace{13mm}\times\exp_{2}(-n(\sum_{i\in A_{l-1}(S)}H(\hat{Y}_{i}|X_{i}))).$ 

\vspace{3mm}
The probability that the incorrect compression estimates $\{\hat{w}_{i}: i\in S\}$ will independently pass each of the typicality checks from $l=0,\ldots,|{\bf L}_{d}|$ is:

\vspace{3mm}
$\exp_{2}(n(\sum^{|{\bf L}_{d}|}_{l=0}H(X_{A_{l}(S)}\hat{Y}_{A_{l-1}(S)}|X_{\tilde{A}_{l}(S)}\hat{Y}_{\tilde{A}_{l}(S)}Y_{d})))$

$\hspace{5mm}\times\exp_{2}(-n(\sum_{i\in S}H(X_{i}\hat{Y}_{i})))$

\vspace{3mm}
Define $\hat{R}_{S}:=\sum_{i\in S}\hat{R}_{i}$.  Since there are $2^{n\hat{R}_{S}}$ possible compression  codewords corresponding to the subset $S$, the probability of error goes to zero if for all $S\subseteq{\cal S}$, the compression rate vector ${\bf\hat{R}}:=(\hat{R}_{2},\ldots,\hat{R}_{d-1})$ satisfies:

\begin{align}
\nonumber
	\hat{R}_{S}&<\sum_{i\in S}H(X_{i}\hat{Y}_{i})\\
	\label{compdecomp3}
	&\hspace{4mm}-\sum^{|{\bf L}_{d}|}_{l=0}H(X_{A_{l}(S)}\hat{Y}_{A_{l-1}(S)}|X_{\tilde{A}_{l}(S)}\hat{Y}_{\tilde{A}_{l-1}(S)}Y_{d}).
\end{align}

%% file: mainresult.tex
\section{Main Result}
\label{mainresult}
Let ${\cal\hat{R}}({\bf L}_{d})$ denote the set of compression rate vectors ${\bf\hat{R}}$ that satisfy (\ref{compdecomp3}) for all $S\subseteq{\cal S}(d)$.
Let ${\cal\hat{R}}_{d}$ denote the set of compression rate vectors that satisfy the following constraint for all $S\subseteq{\cal S}(d)$:
\begin{align}
\label{boundary}
	\hat{R}_{S}&<\sum_{i\in S}H(X_{i}\hat{Y}_{i})-H(X_{S}\hat{Y}_{S}|X_{{\cal S}(d)\setminus S}\hat{Y}_{{\cal S}(d)\setminus S}Y_{d})
\end{align}  
The following theorem is the focus of this paper.
\begin{theorem}
\label{theoremone}
If ${\bf\hat{R}}\in{\cal\hat{R}}_{d}$ then ${\bf\hat{R}}\in{\cal\hat{R}}({\bf L}_{d})$ for some ${\bf L}_{d}$.
\end{theorem}
\begin{proof}
	See Section \ref{prooftheoremone}.
\end{proof}

\begin{figure*}[!t]
        \center{\includegraphics[width=\textwidth]
        {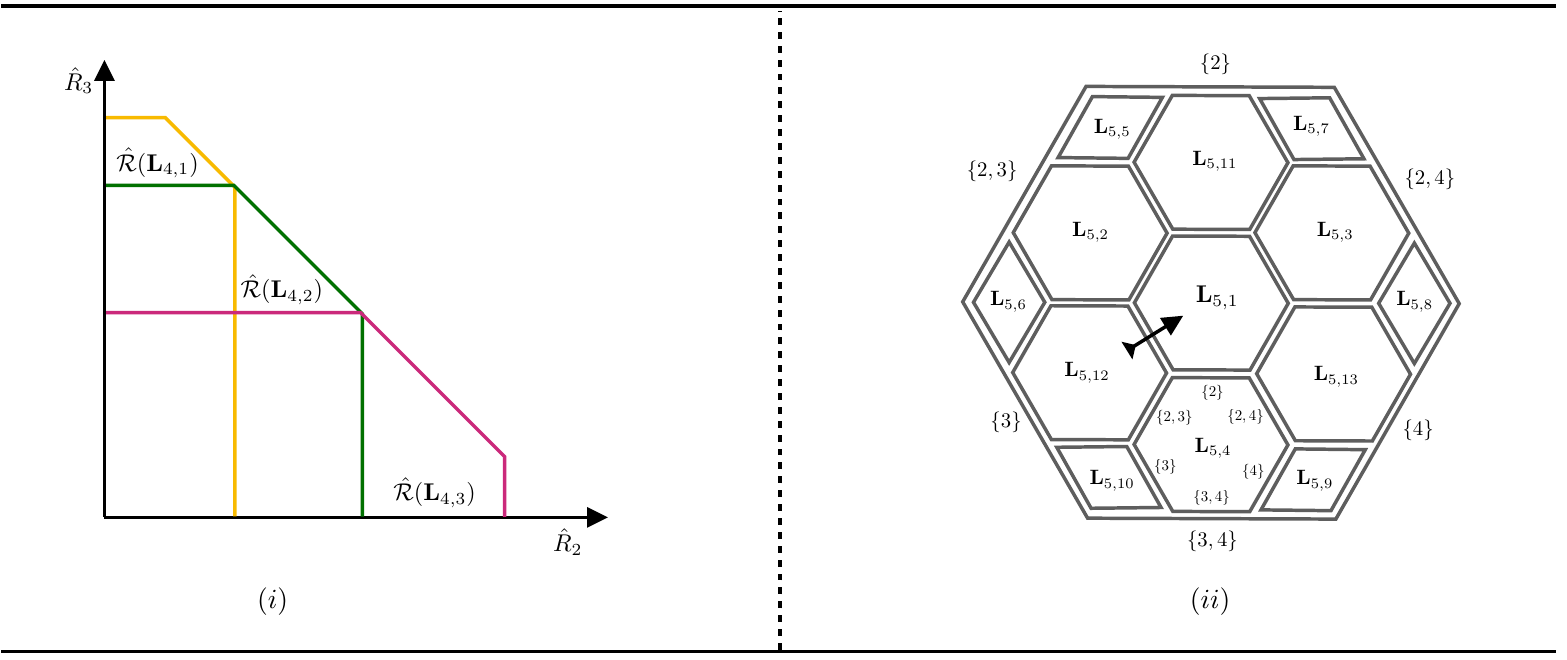}}
        \caption{(i) $\hat{\cal R}_{4}$ for the single-user two-relay channel with relay nodes $\{2,3\}$.  ${\bf L}_{4,1}=(\{2\},\{3\})$, ${\bf L}_{4,2}=(\{2,3\})$, ${\bf L}_{4,3}=(\{3\},\{2\})$ (ii) A 2-D projection of the 3-D region ${\cal\hat{R}}_{5}$ for the single-user three-relay channel with relay nodes $\{2,3,4\}$. ${\bf L}_{5,1}=(\{2,3,4\})$, ${\bf L}_{5,2}=(\{4\},\{2,3\})$, ${\bf L}_{5,3}=(\{3\},\{2,4\})$, ${\bf L}_{5,4}=(\{2\},\{3,4\})$, ${\bf L}_{5,5}=(\{4\},\{3\},\{2\})$, ${\bf L}_{5,6}=(\{4\},\{2\},\{3\})$, ${\bf L}_{5,7}=(\{3\},\{4\},\{2\})$, ${\bf L}_{5,8}=(\{3\},\{2\},\{4\})$, ${\bf L}_{5,9}=(\{2\},\{3\},\{4\})$, ${\bf L}_{5,10}=(\{2\},\{4\},\{3\})$, ${\bf L}_{5,11}=(\{3,4\},\{2\})$, ${\bf L}_{5,12}=(\{2,4\},\{3\})$, ${\bf L}_{5,13}=(\{2,3\},\{4\})$.  Internal facets correspond to (\ref{compdecomp3}) and boundary facets correspond to (\ref{boundary}) for $S=\{2\}, \{2,3\}, \{2,4\}, \{3\}, \{4\},\{2,4\}$.  A shifted layering generates an adjacent subregion: ${\bf L}_{5,1}=\text{\sc shift}({\bf L}_{5,12},\{2,4\})$ where $\{2,4\}$ is the facet of ${\cal\hat{R}}({\bf L}_{5,12})$ that interlocks with ${\cal\hat{R}}({\bf L}_{5,1})$}.
        \label{polytope}
 \end{figure*}
Define $p:=p(x_{1})[\prod_{i\in{\cal S}(d)}p(x_{i})p(\hat{y}_{i}|x_{i}y_{i})]p({\bf y}_{\cal N}|{\bf x}_{\cal N})$.  We have following corollary.    
\begin{corollary}
\label{maincorollary}
Any rate $R$ satisfying: 
	\begin{align}
R<\max_{p}I(X_{1};\hat{Y}_{{\cal S}(d)}Y_{d}|X_{{\cal S}(d)}),		
	\end{align}
 is achievable if for all $S\subseteq{\cal S}(d)$:
	\begin{align}
		\nonumber
		&I(\hat{Y}_{S};Y_{S}|X_{{\cal S}(d)}\hat{Y}_{{\cal S}(d)\setminus S}Y_{d})\\
		\label{achievability}
		&\hspace{25mm}< I(X_{S};\hat{Y}_{{\cal S}(d)\setminus S}|X_{{\cal S}(d)\setminus S}Y_{d}).
	\end{align}
\end{corollary}
\begin{proof}
	For the relays to find compression codewords jointly typical with their observations, the following constraint must be satisfied for all $S\subseteq{\cal S}(d)$:
	\begin{align}
		\label{typpiecomp}
		\sum_{i\in S}I(\hat{Y}_{i};Y_{i}|X_{i})>\hat{R}_{S}.
	\end{align}
	Combining (\ref{typpiecomp}) with (\ref{boundary}) yields:
	\begin{align}
		\nonumber
		\Rightarrow&\sum_{i\in S}I(\hat{Y}_{i};Y_{i}|X_{i})\\
		\label{cor1}
		&\hspace{1mm}<\sum_{i\in S}H(X_{i}\hat{Y}_{i})-H(X_{S}\hat{Y}_{S}|X_{{\cal S}(d)\setminus S}\hat{Y}_{{\cal S}(d)\setminus S}Y_{d}),\\
		\nonumber
		\Rightarrow&\sum_{i\in S}(H(\hat{Y}_{i}|X_{i})-H(\hat{Y}_{i}|X_{i}Y_{i}))\\
		\label{cor2}
		&\hspace{1mm}<\sum_{i\in S}H(X_{i}\hat{Y}_{i})-H(X_{S}\hat{Y}_{S}|X_{{\cal S}(d)\setminus S}\hat{Y}_{{\cal S}(d)\setminus S}Y_{d}),\\
		\nonumber
		\Rightarrow& H(X_{S}\hat{Y}_{S}|X_{{\cal S}(d)\setminus S}\hat{Y}_{{\cal S}(d)\setminus S}Y_{d})-\sum_{i\in S}H(\hat{Y}_{i}|X_{i}Y_{i})\\
		\label{cor3}
		&\hspace{1mm}<\sum_{i\in S}(H(X_{i}\hat{Y}_{i})-H(\hat{Y}_{i}|X_{i})),\\
		\nonumber
		\Rightarrow&H(X_{S}\hat{Y}_{S}|X_{{\cal S}(d)\setminus S}\hat{Y}_{{\cal S}(d)\setminus S}Y_{d})-\sum_{i\in S}H(\hat{Y}_{i}|X_{i}Y_{i})\\
		\label{cor4}
		&\hspace{1mm}<\sum_{i\in S}H(X_{i}),\\
		\nonumber
		\Rightarrow&H(X_{S}\hat{Y}_{S}|X_{{\cal S}(d)\setminus S}\hat{Y}_{{\cal S}(d)\setminus S}Y_{d})-H(\hat{Y}_{S}|X_{S}Y_{S})\\
		\label{cor5}
		&\hspace{1mm}<H(X_{S}),\\
		\nonumber
		\Rightarrow&H(\hat{Y}_{S}|X_{{\cal S}(d)}\hat{Y}_{{\cal S}(d)\setminus S}Y_{d})+H(X_{S}|X_{{\cal S}(d)\setminus S}\hat{Y}_{{\cal S}(d)\setminus S}Y_{d})\\
		\label{cor6}
		&\hspace{1mm}-H(\hat{Y}_{S}|X_{S}Y_{S})<H(X_{S}),\\
		\nonumber
		\Rightarrow&H(\hat{Y}_{S}|X_{{\cal S}(d)}\hat{Y}_{{\cal S}(d)\setminus S}Y_{d})-H(\hat{Y}_{S}|X_{S}Y_{S})\\
		\label{cor7}
		&\hspace{1mm}<H(X_{S})-H(X_{S}|X_{{\cal S}(d)\setminus S}\hat{Y}_{{\cal S}(d)\setminus S}Y_{d}),\\
		\nonumber
		\Rightarrow&H(\hat{Y}_{S}|X_{{\cal S}(d)}\hat{Y}_{{\cal S}(d)\setminus S}Y_{d})-H(\hat{Y}_{S}|X_{{\cal S}(d)}Y_{S}\hat{Y}_{{\cal S}(d)\setminus S}Y_{d})\\
		\label{cor8}
		&\hspace{1mm}<H(X_{S}|X_{{\cal S}(d)\setminus S})-H(X_{S}|X_{{\cal S}(d)\setminus S}\hat{Y}_{{\cal S}(d)\setminus S}Y_{d}),\\
		\nonumber
		\Rightarrow& I(\hat{Y}_{S};Y_{S}|X_{{\cal S}(d)}\hat{Y}_{{\cal S}(d)\setminus S}Y_{d})\\
		\label{cor9}
		&\hspace{28mm}<I(X_{S};\hat{Y}_{{\cal S}(d)\setminus S}Y_{d}|X_{{\cal S}(d)\setminus S}),
	\end{align}
	where (\ref{cor1}) follows from combining (\ref{typpiecomp}) and (\ref{boundary}), (\ref{cor2}) follows from the definition of conditional mutual information, (\ref{cor3}) follows from rearranging, (\ref{cor4}) follows from the chain rule, (\ref{cor5}) follows from the chain rule and because $\hat{Y}_{i}$ is independent conditioned on $(X_{i},Y_{i})$ and $\{X_{i}:i\in{\cal S}(d)\}$ are mutually independent, (\ref{cor6}) follows from the chain rule, (\ref{cor7}) follows by rearranging, (\ref{cor8}) follows because $\hat{Y}_{i}$ is independent conditioned on $(X_{i},Y_{i})$ and $\{X_{i}:i\in{\cal S}(d)\}$ are mutually independent, and (\ref{cor9}) follows from the definition of conditional mutual information.
\end{proof}

%% file: disjoint.tex
 \begin{lemma}
	\label{disjoint}
	For every $S\subseteq{\cal S}(d)$ and $l\in\{0,\ldots,|{\bf L}_{d}|-1\}$, if $S\cap U=\{\}$ then $A_{l}(S)\setminus A_{l}(U)=A^{\prime}_{l}(S)$.
\end{lemma}
\begin{proof}
First suppose $i\in A_{l}(S)\setminus A_{l}(U)$.  The definition in (\ref{A}) implies $\text{\sc layer}(i)=l$ and since $i\notin A_{l}(U)$ it follows from (\ref{shift2}) that $\text{\sc layer}^{\prime}(i)=\text{\sc layer}(i)$. Therefore $\text{\sc layer}^{\prime}(i)=l$.  The definition in (\ref{A}) also implies $i\in S\setminus U\subseteq S$.  By the same definition and $\text{\sc layer}^{\prime}(i)=l$, it follows that $i\in A^{\prime}_{l}(S)$.  

Now suppose $i\in A^{\prime}_{l}(S)$.  The definition in (\ref{A}) implies $i\in S$.  Since $S\cap U=\{\}$, it follows that $i\notin A_{l}(U)$.  Therefore $\text{\sc layer}^{\prime}(i)=\text{\sc layer}(i)=l$ and $i\in A_{l}(S)\setminus A_{l}(U)$. 
\end{proof}

%% file: remain.tex
\begin{lemma}
\label{remain}
For every $S\subseteq{\cal S}(d)$ and $l\in\{0,\ldots,|{\bf L}_{d}|-1\}$, $A^{\prime}_{l}(S\setminus U)=A_{l}(S\setminus U)$
\end{lemma}
\begin{proof}
Fix $i\in A^{\prime}_{l}(S\setminus U)$.  The definition in (\ref{A}) implies $i\in S\setminus U$.  Since $i\notin U$, (\ref{A}) implies $i\notin A_{l}(U)$.  It follows from (\ref{shift2}), that $\text{\sc layer}^{\prime}(i)=\text{\sc layer}(i)=l$.  Therefore (\ref{A}) implies $i\in A_{l}(S\setminus U)$.  Now fix $i\in A_{l}(S\setminus U)$.  Again (\ref{A}) implies $i\in S\setminus U$ so $i\notin A_{l}(U)$.  Then (\ref{shift2}) implies $\text{\sc layer}^{\prime}(i)=\text{\sc layer}(i)=l$, so that $i\in A^{\prime}_{l}(S\setminus U)$.  
\end{proof}

%% file: join.tex
\begin{lemma}
\label{join}
For every $S\subseteq{\cal S}(d)$ and $l\in\{0,\ldots,|{\bf L}_{d}|-1\}$, $A_{l-1}(S\cap U)=A^{\prime}_{l}(S\cap U)$
\end{lemma}
\begin{proof}
Fix $i\in A_{l-1}(S\cap U)$.  The definition in (\ref{A}) implies $\text{\sc layer}(i)=l-1$ and $i\in S\cap U$.  Since $i\in U$, (\ref{A}) implies $i\in A_{l-1}(U)$.  It follows from (\ref{shift2}), that $\text{\sc layer}^{\prime}(i)=\text{\sc layer}(i)+1=l$.  Therefore (\ref{A}) implies $i\in A^{\prime}_{l}(S\cap U)$.  Now fix $i\in A^{\prime}_{l}(S\cap U)$.  Again (\ref{A}) implies $i\in S\cap U$.  Then (\ref{shift2}) implies $\text{\sc layer}(i)=\text{\sc layer}^{\prime}(i)-1=l-1$, so that $i\in A_{l-1}(S\cap U)$.
\end{proof}

%% file: induction.tex
\begin{lemma}
	\label{induction}
	$Z^{\prime}=({\cal S}(d)\setminus U)\cup Z$.
\end{lemma}
\begin{proof}
	The proof is by contradiction.  Suppose there is some $S\subseteq({\cal S}(d)\setminus U)\cup V$ that violates (\ref{compdecomp2}) for ${\bf L}^{\prime}_{d}$.  By assumption,
	\begin{align}
		\label{lemmaVprep2}
		\hat{R}_{S}&>\displaystyle\sum_{i\in S}H(X_{i}\hat{Y}_{i})-\displaystyle\sum^{|{\bf L}_{d}^{\prime}|}_{l=0}h[A^{\prime}_{l}(S)|\tilde{A}^{\prime}_{l}(S)]
	\end{align}
	{\bf Case 1:}  $S\cap U=\{\}$.
	We will prove the following inequalities:
	\begin{align}
		\nonumber
		\hat{R}_{U\cup S}&>\displaystyle\sum_{i\in U}H(X_{i}\hat{Y}_{i})-\displaystyle\sum^{|{\bf L}_{d}|}_{l=0}h[A_{l}(U)|\tilde{A}_{l}(U)]\\
		\label{lemmaU1}
		&\hspace{8mm}+\displaystyle\sum_{i\in S}H(X_{i}\hat{Y}_{i})-\displaystyle\sum^{|{\bf L}^{\prime}_{d}|}_{k=0}h[A^{\prime}_{k}(S)|\tilde{A}^{\prime}_{k}(S)]\\
	    \nonumber
		&=\displaystyle\sum_{i\in U\cup S}H(X_{i}\hat{Y}_{i})-\displaystyle\sum^{|{\bf L}_{d}|}_{l=0}h[A_{l}(U)|\tilde{A}_{l}(U)]\\
		\label{lemmaU2}
		&\hspace{8mm}-\displaystyle\sum^{|{\bf L}^{\prime}_{d}|}_{k=0}h[A_{k}(S)\setminus A_{k}(U)|\tilde{A}_{k}(S)\setminus A_{k}(U)]\\
	    \nonumber
		&=\displaystyle\sum_{i\in U\cup S}H(X_{i}\hat{Y}_{i})-\displaystyle\sum^{|{\bf L}_{d}|}_{l=0}(h[A_{l}(U)|\tilde{A}_{l}(U)]\\
		\label{lemmaU3}
		&\hspace{8mm}+h[A_{l}(S)\setminus A_{l}(U)|\tilde{A}_{l}(S)\setminus A_{l}(U)])\\
		\nonumber
		&=\displaystyle\sum_{i\in U\cup S}H(X_{i}\hat{Y}_{i})\\
		\label{lemmaU3a}
		&\hspace{8mm}-\displaystyle\sum^{|{\bf L}_{d}|}_{l=0}h[A_{l}(U)\cup A_{l}(S)|\tilde{A}_{l}(S)\setminus A_{l}(U)]\\
		\nonumber
		&=\displaystyle\sum_{i\in U\cup S}H(X_{i}\hat{Y}_{i})\\
		\label{lemmaU3b}
		&\hspace{8mm}-\displaystyle\sum^{|{\bf L}_{d}|}_{l=0}h[A_{l}(U\cup S)|\tilde{A}_{l}(U\cup S)]
	\end{align}
which contradicts the assumption that $U$ is the largest subset that violates (\ref{compdecomp}) since $S$ and $U$ are disjoint.  To justify (\ref{lemmaU1})-(\ref{lemmaU3b}), we introduce some preliminary claims.  

\hspace{-3.5mm}{\bf Claim \ref{induction}.1}  \textit{If $S_{1}\subseteq{\cal N}$, $S_{3}\subseteq S_{2}\subseteq{\cal N}$ and $S_{1}\cap S_{2}=\{\}$ then:}
\begin{align}
	\label{claim1}
	h[S_{1}|S_{2}]+h[S_{3}|S_{2}\setminus S_{3}]&=h[S_{1}\cup S_{3}|S_{2}\setminus S_{3}].
\end{align}
\begin{proof}
	First, we verify the terms in (\ref{claim1}) are consistent with the definition in (\ref{hdef}).  For the first term in (\ref{claim1}), set $B_{l}:=S_{1}$ and $C_{l}:=S_{1}\cup S_{2}$.  For the second term, set $B_{l}:=S_{3}$ and $C_{l}:=S_{2}$.  For the third term, set $B_{l}:=S_{1}\cup S_{3}$ and $C_{l}:=S_{1}\cup S_{2}$.  In all cases $B_{l}\subseteq C_{l}$ and $\tilde{B}_{l}=C_{l}\setminus B_{l}$ so the terms in (\ref{claim1}) are well-defined.  To prove (\ref{claim1}), apply the chain rule.
\end{proof}
\hspace{-3.5mm}{\bf Claim \ref{induction}.3} \textit{For any $S\subseteq{\cal S}(d)$ and every $k\in 0,\ldots,|{\bf L}^{\prime}_{d}|-1$:}
\begin{align}
\label{induction8}
&\hspace{-2mm}h[A^{\prime}_{k}(S)|\tilde{A}^{\prime}_{k}(S)]=h[A_{k}(S)\setminus A_{k}(U)|\tilde{A}^{\prime}_{k}(S)\setminus A_{k}(U)]
\end{align}	
\begin{proof}
First we prove that:
\begin{align}
	\label{subclaim}
	\{\cup^{k}_{q=0}L_{q}\}\setminus A_{k}(U)&=\{\cup^{k}_{q=0}L^{\prime}_{q}\}.
\end{align}
To prove (\ref{subclaim}), fix $i\in\{\cup^{k}_{q=0}L_{q}\}\setminus A_{k}(U)$.  If $i\in\{\cup^{k-1}_{q=0}L_{q}\}$, then (\ref{shift2}) implies $i\in\{\cup^{k}_{q=0}L^{\prime}_{q}\}$.  If $i\in L_{k}\setminus A_{k}(U)$, then (\ref{shift2}) implies $i\in L^{\prime}_{k}\subseteq\cup^{k}_{q=0}L^{\prime}_{q}$.  Now fix $i\in\{\cup^{k}_{q=0}L^{\prime}_{q}\}$.  For any $0\leq q\leq k$, if $i\in L^{\prime}_{q}$ then (\ref{shift2}) implies $i\in L_{q-1}\cup L_{q}$ and $i\notin A_{q}(U)$.  Therefore, $i\in\{\cup^{k}_{q=0}L_{q}\}\setminus A_{k}(U)$.  Next, if $S\cap U=\{\}$, we prove that:
\begin{align}
	\label{subclaim2}
	A^{\prime}_{k}(S)=A_{k}(S).
\end{align}
To prove (\ref{subclaim2}), suppose $i\in A_{k}(S)$ but $i\notin A^{\prime}_{k}(S)$.  By the definition in (\ref{aln}), $i\in S$.  Moreover, (\ref{shift2}) implies $i\in A_{k}(U)$ which implies $i\in U$.  But $S\cap U=\{\}$ which is a contradiction.  Similarly, suppose $i\in A^{\prime}_{k}(S)$, but $i\notin A_{k}(S)$.  By definition in (\ref{aln}), $i\in S$.  Furthermore, (\ref{shift2}) implies $i\in A_{l-1}(U)$, which also means $i\in U$.  But again, $S\cap U=\{\}$ which is a contradiction.  Now observe that:
\begin{align}
	\label{claim3a}
	\tilde{A}^{\prime}_{k}(S)&:=\{\cup^{k}_{q=0}L^{\prime}_{q}\}\setminus A^{\prime}_{k}(S),\\
	\label{claim3b}
	&=\{\{\cup^{k}_{q=0}L_{q}\}\setminus A_{k}(U)\}\setminus A^{\prime}_{k}(S),\\
	\label{claim3b1}
	&=\{\{\cup^{k}_{q=0}L_{q}\}\setminus A_{k}(U)\}\setminus A_{k}(S),\\
	\label{claim3b2}
	&=\{\cup^{k}_{q=0}L_{q}\}\setminus \{A_{k}(U)\cup A_{k}(S)\},\\
	\label{claim3c}
	&=\{\cup^{k}_{q=0}L_{q}\}\setminus A_{k}(S)\}\setminus A_{k}(U)\},\\
	\label{claim3c1}
	&=\tilde{A}_{k}(S)\setminus A_{k}(U),
\end{align}
where (\ref{claim3a}) follows from the definition in (\ref{Atilde}), (\ref{claim3b}) follows from substituting (\ref{subclaim}) in (\ref{claim3a}), (\ref{claim3b1}) follows from substituting (\ref{subclaim2}) in (\ref{claim3b}), (\ref{claim3b2}) and (\ref{claim3c}) follow from normal set-theoretic operations, and (\ref{claim3c1}) follows from (\ref{claim3c}) and the definition in (\ref{Atilde}).  We have the following:
\begin{align}
	\label{claim3d}
	\hspace{-2mm}h[A^{\prime}_{k}(S)|\tilde{A}^{\prime}_{k}(S)]&=h[A_{k}(S)\setminus A_{k}(U)|\tilde{A}^{\prime}_{k}(S)],\\
	\label{claim3e}
	&=h[A_{k}(S)\setminus A_{k}(U)|\tilde{A}_{k}(S)\setminus A_{k}(U)]
\end{align}
where (\ref{claim3d}) follows from Lemma \ref{disjoint}, and (\ref{claim3e}) follows from (\ref{claim3c1}), proving (\ref{induction8}).
\end{proof}
\hspace{-3.5mm}{\bf Claim \ref{induction}.4} \textit{For any $S\subseteq{\cal S}(d)$ and every $l\in 0,\ldots,|{\bf L}_{d}|-1$:}
\begin{align}
\nonumber
&\hspace{-12mm}h[A_{l}(U\cup S)|\tilde{A}_{l}(U\cup S)]\\
\nonumber
&\hspace{-12mm}=h[A_{l}(U)|\tilde{A}_{l}(U)]\\
\label{induction9}
&\hspace{5mm}+h[A_{l}(S)\setminus A_{l}(U)|\tilde{A}_{l}(S)\setminus A_{l}(U)]
\end{align}
\begin{proof}
To prove (\ref{induction9}), we invoke Claim 6.1.  First, we verify that the premises of Claim 6.1 are satisfied.  Let $S_{1}:=A_{l}(U)$, $S_{2}:=\tilde{A}_{l}(U)$, and $S_{3}:=A_{l}(S)\setminus A_{l}(U)$.  From (\ref{Atilde}):
\begin{align}
	\label{induction8b}
	\tilde{A}_{l}(U)=(\cup^{l}_{q=0}L_{q})\setminus A_{l}(U).
\end{align}
Since $A_{l}(S)\subseteq L_{l}$, (\ref{induction8b}) implies $S_{3}\subseteq S_{2}$.  By inspection, $S_{1}\cap S_{2}=\{\}$.  Finally,
\begin{align}
	\nonumber
	&\hspace{-3mm}S_{2}\setminus S_{3}\\
	\nonumber
	&\hspace{-3mm}=\tilde{A}_{l}(U)\setminus (A_{l}(S)\setminus A_{l}(U)),\\
	\nonumber
	&\hspace{-3mm}=((\cup^{l}_{q=0}L_{q})\setminus A_{l}(U))\setminus(A_{l}(S)\setminus A_{l}(U)),\\
	\nonumber
	&\hspace{-3mm}=(\cup^{l}_{q=0}L_{q})\setminus(A_{l}(S)\cup A_{l}(U)),\\
	\nonumber
	&\hspace{-3mm}=((\cup^{l}_{q=0}L_{q})\setminus A_{l}(S))\setminus A_{l}(U),\\
	\label{induction8d}
	&\hspace{-3mm}=\tilde{A}_{l}(S)\setminus A_{l}(U),
\end{align}
where the equalities follow from normal set-theoretic operations and (\ref{induction8b}).  Furthermore,
	\begin{align}
		\nonumber
		S_{1}\cup S_{3}&=A_{l}(U)\cup(A_{l}(S)\setminus A_{l}(U))\\
		\label{induction8c}
		&=A_{l}(U)\cup A_{l}(S).
	\end{align}
	The premises of Claim 6.1 are satisfied in (\ref{induction9}),  (\ref{induction8d}) and  (\ref{induction8c}). 
\end{proof}	
We are now ready to verify (\ref{lemmaU1})-(\ref{lemmaU3b}).  To justify (\ref{lemmaU1}), we invoke the definition of $U$ as the largest subset of ${\cal S}(d)$ that violates (\ref{compdecomp2}) and the fact that $S$ satisfies (\ref{lemmaVprep2}) by hypothesis.  Moreover (\ref{lemmaU2}) follows from Claim \ref{induction}.3, and (\ref{lemmaU3}) follows from relabeling and the fact that $|{\bf L}^{\prime}_{d}|\leq|{\bf L}_{d}|+1$ and $A_{l}(S)=\{\}$ for $l=|{\bf L}_{d}|$.  To justify (\ref{lemmaU3a}), we invoke Claim \ref{induction}.4.  To justify (\ref{lemmaU3b}), we have the following equalities:
 In addition:
\begin{align}
	\nonumber
	\tilde{A}_{l}(S)\setminus A_{l}(U)&=(\cup^{l}_{q=0}L_{q})\setminus A_{l}(S)\setminus A_{l}(U),\\
	\nonumber
	&=(\cup^{l}_{q=0}L_{q})\setminus(A_{l}(S)\cup A_{l}(U)),\\
	\nonumber
	&=(\cup^{l}_{q=0}L_{q})\setminus A_{l}(S\cup U),\\
	\label{induction8e}
	&=\tilde{A}_{l}(S\cup U),
\end{align} 
where the equalities follow from normal set-theoretic operations.  Applying (\ref{induction8e}) to (\ref{lemmaU3a}) gives (\ref{lemmaU3b}), which proves Lemma \ref{induction} for Case 1.

{\bf Case 2:}  $S\cap U\neq\{\}$.  Since $\{S\setminus U\}\cap U=\{\}$, Case 1  implies:
\begin{align}
	\nonumber
	R_{S\setminus U}&<\displaystyle\sum_{i\in S\setminus U}H(X_{i}\hat{Y}_{i})\\
	\label{lemmaVprep1}
	&\hspace{18mm}-\displaystyle\sum^{|{\bf L}^{\prime}_{d}|}_{l=0}h[A^{\prime}_{l}(S\setminus U)|\tilde{A}^{\prime}_{l}(S\setminus U)]
\end{align}
We will prove the following series of inequalities:
\begin{align}
		\nonumber
		R_{S\cap U}&>\displaystyle\sum_{i\in S\cap U}H(X_{i}\hat{Y}_{i})\\
		\label{lemmaV1}
		&\hspace{-8mm}-\sum^{|{\bf L}^{\prime}_{d}|}_{l=0}h[A^{\prime}_{l}(S\cap U)\setminus A^{\prime}_{l}(S\setminus U)|\tilde{A}^{\prime}_{l}(S\cap U)\setminus A^{\prime}_{l}(S\setminus U)]\\
		\label{lemmaV2}
		&\hspace{-8mm}\geq\displaystyle\sum_{i\in S\cap U}H(X_{i}\hat{Y}_{i})-\displaystyle\sum^{|{\bf L}_{d}|}_{l=0}h[A_{l}(S\cap U)|\tilde{A}_{l}(S\cap U)].
\end{align}
Since $S\subseteq({\cal S}(d)\setminus U)\cup Z$ by hypothesis, it follows that $(S\cap U)\subseteq Z$.  Therefore (\ref{lemmaV2}) contradicts the assumption that all subsets of $Z$ satisfy (\ref{compdecomp2}) for ${\bf L}_{d}$.  To justify (\ref{lemmaV1})-(\ref{lemmaV2}), we introduce some preliminary claims.

\hspace{-3.5mm}{\bf Claim \ref{induction}.5}   \textit{For any $S\subseteq{\cal S}(d)$:}
\begin{align}
	\nonumber
		R_{S\cap U}&>\displaystyle\sum_{i\in S\cap U}H(X_{i}\hat{Y}_{i})\\
		\nonumber
		&\hspace{-8mm}-\sum^{|{\bf L}^{\prime}_{d}|}_{l=0}h[A^{\prime}_{l}(S\cap U)\setminus A^{\prime}_{l}(S\setminus U)|\tilde{A}^{\prime}_{l}(S\cap U)\setminus A^{\prime}_{l}(S\setminus U)].
\end{align}

\begin{proof}
Fix any $l\in\{0,\ldots,|{\bf L}^{\prime}_{d}|-1\}$ and consider the following sequence of equalities:
\begin{align}
\nonumber
&\hspace{-2.4mm}h[A^{\prime}_{l}(S\setminus U)|\tilde{A}^{\prime}_{l}(S\setminus U)]\\
   \nonumber
	&\hspace{-2mm}+h[A^{\prime}_{l}(S\cap U)\setminus A^{\prime}_{l}(S\setminus U)|\tilde{A}^{\prime}_{l}(S\cap U)\setminus A^{\prime}_{l}(S\setminus U)]\\
	\label{induction10}
	&\hspace{-2mm}=h[A^{\prime}_{l}(S\setminus U)\cup A^{\prime}_{l}(S\cap U)|\tilde{A}^{\prime}_{l}(S\cap U)\setminus A^{\prime}_{l}(S\setminus U)],\\
	\label{induction11}
	&\hspace{-2mm}= h[A^{\prime}_{l}(S)|\tilde{A}^{\prime}_{l}(S)].
\end{align}
To justify (\ref{induction10}), we invoke Claim \ref{induction}.1.  First, we verify that  the premises of Claim \ref{induction}.1 are satisfied.  Let $S_{1}:=A^{\prime}_{l}(S\setminus U)$, $S_{2}:=\tilde{A}^{\prime}_{l}(S\setminus U)$, and $S_{3}:=A^{\prime}_{l}(S\cap U)\setminus A^{\prime}_{l}(S\setminus U)$.  By inspection $S_{1}\cap S_{2}=\{\}$.  The definition in (\ref{A}) implies $A^{\prime}_{l}(S\cap U)\subseteq L^{\prime}_{l}$ and (\ref{Atilde}) implies $\tilde{A}^{\prime}_{l}(S\setminus U):=(\cup^{l}_{q=0}L^{\prime}_{q})\setminus A^{\prime}_{l}(S\setminus U)$.  It follows that: 
\begin{align}
	\label{induction16}
	A^{\prime}_{l}(S\cap U)\setminus A^{\prime}_{l}(S\setminus U)\subseteq\tilde{A}^{\prime}_{l}(S\setminus U),
\end{align}
which implies $S_{3}\subseteq S_{2}$.  Moreover,  
\begin{align}
	\nonumber
	&\hspace{-9mm}S_{2}\setminus S_{3}\\
	\nonumber
	&\hspace{3mm}\hspace{-9mm}=\tilde{A}^{\prime}_{l}(S\setminus U)\setminus (A^{\prime}_{l}(S\cap U)\setminus A^{\prime}_{l}(S\setminus U))\\
	\nonumber
	&\hspace{3mm}\hspace{-9mm}=((\cup^{l}_{q=0}L^{\prime}_{q})\setminus A^{\prime}_{l}(S\setminus U))\\
	\nonumber
	&\hspace{3mm}\hspace{-9mm}\hspace{30mm}\setminus(A^{\prime}_{l}(S\cap U)\setminus A^{\prime}_{l}(S\setminus U)),\\
	\nonumber
	&\hspace{3mm}\hspace{-9mm}=(\cup^{l}_{q=0}L^{\prime}_{q})\setminus(A^{\prime}_{l}(S\setminus U)\cup A^{\prime}_{l}(S\cap U)),\\
	\nonumber
	&\hspace{3mm}\hspace{-9mm}=((\cup^{l}_{q=0}L^{\prime}_{q})\setminus A^{\prime}_{l}(S\cap U))\setminus A^{\prime}_{l}(S\setminus U),\\
	\label{induction17}
	&\hspace{3mm}\hspace{-9mm}=\tilde{A}^{\prime}_{l}(S\cap U)\setminus A^{\prime}_{l}(S\setminus U),
\end{align}
where the equalities follow from normal set-theoretic operations and (\ref{Atilde}).  Furthermore, 
\begin{align}
	\nonumber
	S_{1}\cup S_{3}&=A^{\prime}_{l}(S\setminus U)\cup(A^{\prime}_{l}(S\cap U)\setminus A^{\prime}_{l}(S\setminus U))\\
	\label{induction18}
	&=A^{\prime}_{l}(S\cap U)\cup A^{\prime}_{l}(S\setminus U).
\end{align}
Since the premises of Claim \ref{induction}.1 are satisfied in (\ref{induction16}), (\ref{induction17}), and (\ref{induction18}), invoking Claim \ref{induction}.1 proves (\ref{induction10}).  To prove (\ref{induction11}) observe:
\begin{align}
\nonumber
&\hspace{-2mm}\tilde{A}^{\prime}_{l}(S\cap U)\setminus A^{\prime}_{l}(S\setminus U)\\
\label{induction12}
&\hspace{14mm}=((\cup^{l}_{q=0}L^{\prime}_{q})\setminus A^{\prime}_{l}(S\cap U))\setminus A^{\prime}_{l}(S\setminus U),\\
\label{induction13}
&\hspace{14mm}=(\cup^{l}_{q=0}L^{\prime}_{q})\setminus(A^{\prime}_{l}(S\cap U)\cup A^{\prime}_{l}(S\setminus U)),\\
\label{postinduction13}
&\hspace{14mm}=\tilde{A}^{\prime}_{l}(S),
\end{align}
where (\ref{induction12}) follows from (\ref{Atilde}) and (\ref{induction13}) follows from normal set-theoretic operations.  The definition in (\ref{A}) implies that $i\in A^{\prime}_{l}(S)$ if and only if $i\in A^{\prime}_{l}(S\setminus U)\cup A^{\prime}_{l}(S\cap U)$.  Therefore $A^{\prime}_{l}(S)=A^{\prime}_{l}(S\setminus U)\cup A^{\prime}_{l}(S\cap U)$ which implies (\ref{postinduction13}). Therefore (\ref{induction11}) follows from (\ref{postinduction13}).  Since  $R_{S}=R_{S\setminus U}+R_{S\cap U}$, Claim \ref{induction}.4 follows from (\ref{lemmaVprep2}), (\ref{lemmaVprep1}), and  (\ref{induction11}).
\end{proof} 
Two more preliminary claims are required to justify (\ref{lemmaV2}).

 \hspace{-3.5mm}{\bf Claim \ref{induction}.6}   \textit{For any $S\subseteq{\cal S}(d)$ and every $l=0,\ldots,|{\bf L}^{\prime}_{d}|-1$:}
\begin{align}
\nonumber
	&\hspace{-2mm}h[A^{\prime}_{l}(S\cap U)\setminus A^{\prime}_{l}(S\setminus U)|\tilde{A}^{\prime}_{l}(S\cap U)\setminus A^{\prime}_{l}(S\setminus U)]=\\
	\label{induction19}
	&\hspace{15mm}h[A_{l-1}(S\cap U)|\tilde{A}^{\prime}_{l}(S\setminus U)\setminus A_{l-1}(S\cap U)]
\end{align}
\begin{proof}
First we show that $A_{l-1}(S\cap U)=A^{\prime}_{l}(S\cap U)\setminus A^{\prime}_{l}(S\setminus U)$.  Lemma \ref{join} implies that $A_{l-1}(S\cap U)=A^{\prime}_{l}(S\cap U)$.  Moreover, Lemma \ref{remain} implies $A^{\prime}_{l}(S\setminus U)\subseteq A_{l}(S\setminus U)$.  Now (\ref{A}) implies $A_{l-1}(S\cap U)\subseteq L_{l-1}$ and $A_{l}(S\setminus U)\subseteq L_{l}$.  Furthermore, (L2) implies $L_{l-1}\cap L_{l}=\{\}$.  It follows that $A_{l-1}(S\cap U)\cap A^{\prime}_{l}(S\setminus U)=\{\}$.  Therefore $A_{l-1}(S\cap U)=A^{\prime}_{l}(S\cap U)\setminus A^{\prime}_{l}(S\setminus U)$.  

Next we show that $\tilde{A}^{\prime}_{l}(S\cap U)\setminus A^{\prime}_{l}(S\setminus U)=\tilde{A}^{\prime}_{l}(S\setminus U)\setminus A_{l-1}(S\cap U)$.  Observe that: 
\begin{align}
\nonumber
&\tilde{A}^{\prime}_{l}(S\cap U)\setminus A^{\prime}_{l}(S\setminus U)\\
\label{insert1}
&\hspace{7mm}:=((\cup^{l}_{q=0}L^{\prime}_{q})\setminus(A^{\prime}_{l}(S\cap U))\setminus A^{\prime}_{l}(S\setminus U),\\
\label{insert2}
&\hspace{7mm}=(\cup^{l}_{q=0}L^{\prime}_{q})\setminus(A^{\prime}_{l}(S\cap U)\cup A^{\prime}_{l}(S\setminus U)),\\
\label{insert3}
&\hspace{7mm}=(\cup^{l}_{q=0}L^{\prime}_{q})\setminus(A_{l-1}(S\cap U)\cup A^{\prime}_{l}(S\setminus U)),\\
\label{insert4}
&\hspace{7mm}=((\cup^{l}_{q=0}L^{\prime}_{q})\setminus A^{\prime}_{l}(S\setminus U))\setminus(A_{l-1}(S\cap U),\\
\label{insert5}
&\hspace{7mm}=\tilde{A}^{\prime}_{l}(S\setminus U)\setminus A_{l-1}(S\cap U),
\end{align}
where (\ref{insert1}) follows from the definition in (\ref{Atilde}), (\ref{insert2}) follows from normal set-theoretic operations, (\ref{insert3}) follows from Lemma \ref{join}, (\ref{insert4}) follows from normal set-theoretic operations, and (\ref{insert5}) follows from the definition in (\ref{Atilde}), thus proving the claim.  
\end{proof} 

 \hspace{-3.5mm}{\bf Claim \ref{induction}.7}   \textit{For any $S\subseteq{\cal S}(d)$ and every $l=0,\ldots,|{\bf L}^{\prime}_{d}|-1$:}
\begin{align}
	\nonumber
	&\hspace{-5mm}h[A_{l-1}(S\cap U)|\tilde{A}^{\prime}_{l}(S\setminus U)\setminus A_{l-1}(S\cap U)]\geq\\
	&\hspace{25mm}h[A_{l-1}(S\cap U)|\tilde{A}_{l-1}(S\cap U)].
\end{align}
\begin{proof}
Consider the following sequence of inequalities:
\begin{align}
	\nonumber
	&\hspace{-2mm}\tilde{A}^{\prime}_{l}(S\setminus U)\setminus A_{l-1}(S\cap U)\\
	\label{induction21}
	&\hspace{1mm}=((\cup^{l}_{q=0}L^{\prime}_{q})\setminus A^{\prime}_{l}(S\setminus U))\setminus A_{l-1}(S\cap U),\\
	\label{induction22}
	&\hspace{1mm}=(\cup^{l}_{q=0}L^{\prime}_{q})\setminus (A^{\prime}_{l}(S\setminus U)\cup A_{l-1}(S\cap U)),\\
	\nonumber
	&\hspace{1mm}=((\cup^{l}_{q=0}L_{q})\setminus A_{l}(U))\\
	\label{induction23}
	&\hspace{29mm}\setminus(A^{\prime}_{l}(S\setminus U)\cup A_{l-1}(S\cap U)),\\
	\nonumber
	&\hspace{1mm}=(\cup^{l}_{q=0}L_{q})\\
	\label{induction24}
	&\hspace{17mm}\setminus(A_{l}(S\setminus U)\cup A_{l-1}(S\cap U)\cup A_{l}(U)),\\
	\label{induction25}
	&\hspace{1mm}\supseteq (\cup^{l-1}_{q=0}L_{q})\setminus A_{l-1}(S\cap U),\\
	\label{induction26}
	&\hspace{1mm}=\tilde{A}_{l-1}(S\cap U),
\end{align}
where (\ref{induction21}) follows from the definition in (\ref{Atilde}), (\ref{induction22}) follows from standard set-theoretic operations, (\ref{induction23}) follows by substituting (\ref{subclaim}) into (\ref{induction22}), (\ref{induction24}) follows because Lemma \ref{remain} implies $A^{\prime}_{l}(S\setminus U)=A_{l}(S\setminus U)$, (\ref{induction25}) follows because (\ref{A}) implies $A_{l}(U)\cup A_{l}(S\setminus U)\subseteq L_{l}$, and (\ref{induction26}) follows from the definition of $\tilde{A}_{l}(\hspace{1mm}\cdot\hspace{1mm})$ in (\ref{Atilde}).  Since removing independent conditional random variables reduces the mutual information, Claim \ref{induction}.7 follows from (\ref{induction26}).
\end{proof}
Combining Claim \ref{induction}.6 and Claim \ref{induction}.7 shows that  (\ref{lemmaV1}) implies:
\begin{align}
	\nonumber
	&\hspace{-3mm}R_{S\cap U}>\displaystyle\sum_{i\in S\cap U}H(X_{i}\hat{Y}_{i})\\
	\label{induction27}
	&\hspace{16mm}+\displaystyle\sum^{|{\bf L}^{\prime}_{d}|-1}_{l=0}h[A_{l-1}(S\cap U)|\tilde{A}_{l-1}(S\cap U)].
\end{align}
Observe that (\ref{shift2}) implies $|{\bf L}_{d}|\leq|{\bf L}^{\prime}_{d}|\leq|{\bf L}_{d}|+1$.  Since $ l\leq|{\bf L}_{d}^{\prime}|-1$ it follows that $l-1\leq|{\bf L}^{\prime}_{d}|-2\leq|{\bf L}_{d}|-1$.  Substituting $l^{\prime}=l-1$ into (\ref{induction27}) and relabelling $l^{\prime}=l$ yields (\ref{lemmaV2}) which completes the proof of Lemma \ref{induction}.
\end{proof}

Given a sequence of layerings $\{{\bf L}_{d,n}:n\in\mathbb{N}\}$, let ${\bf L}_{d,n}:=(L_{0},L_{1},\ldots,L_{|{\bf L}_{d,n}|-1})$ be a layering that satisfies (L1)-(L5).  For every $S\subseteq{\cal S}(d)$, let:
\begin{align}
	\label{aln}
	A_{l,n}(S)&:=S\cap L_{l,n},\\
	\label{alntilde}
	\tilde{A}_{l,n}(S)&:=(\cup^{l}_{q=0}L_{q,n})\setminus A_{l,n}(S),
\end{align}
For every $S\subseteq{\cal S}(d)$,  let ${\cal\hat{R}}({\bf L}_{d,n})$ denote the set of rate vectors ${\bf\hat{R}}$ that satisfy:  
\begin{align}
	\nonumber
	\hspace{-2.2mm}\hat{R}_{S}&<\sum_{i\in S}H(X_{i}\hat{Y}_{i})\\
	\label{precompdecomp}
	&\hspace{-4mm}-\sum^{|{\bf L}_{d,n}|}_{l=0}H(X_{A_{l,n}(S)}\hat{Y}_{A_{l-1,n}(S)}|X_{\tilde{A}_{l,n}(S)}\hat{Y}_{\tilde{A}_{l-1,n}(S)}Y_{d}),\\
	\label{compdecomp}
	&=\sum_{i\in S}H(X_{i}\hat{Y}_{i})-\sum^{|{\bf L}_{d,n}|}_{l=0}h[A_{l,n}(S)|\tilde{A}_{l,n}(S)],
\end{align}   
where (\ref{compdecomp}) follows from (\ref{precompdecomp}) and (\ref{hdef}).  For some fixed ${\bf\hat{R}}\in{\cal\hat{R}}_{d}$, let $U_{n}\subset{\cal S}(d)$ denote the largest set that violates (\ref{compdecomp}) with respect to ${\bf L}_{d,n}$ and ${\bf\hat{R}}$, and $Z_{n}\subseteq{\cal S}(d)$ denote the set in which all subsets $S\subseteq Z_{n}$ satisfy (\ref{compdecomp}) with respect to ${\bf L}_{d,n}$ and ${\bf\hat{R}}$.

Let ${\bf L}_{d,n+1}:=\text{\sc shift}({\bf L}_{d,n},U_{n})$, and define the $\text{\sc shift}(\cdot\hspace{1mm},\cdot)$ operator as follows.  Let $\text{\sc layer}_{n}(\cdot)$ correspond to the layering ${\bf L}_{d,n}$, and for every $i\in {\cal S}(d)$:
\begin{align}
\label{shift4}
\text{\sc layer}_{n+1}(i)=\begin{cases}l&i\in A_{l,n}({\cal N})\setminus A_{l,n}(U_{n}),\\ l+1 & i\in A_{l,n}(U_{n}).\end{cases}	
\end{align}
For every $n\in\mathbb{N}$, the pair $(U_{n},Z_{n})$ satisfies:
\begin{align}
	\label{inductionupdate}
	Z_{n+1}=({\cal S}(d)\setminus U_{n})\cup Z_{n},
\end{align} 
where (\ref{inductionupdate}) follows from Lemma \ref{induction}.

%% file: lemma7.tex
\begin{lemma}
	\label{grandefinale}
	${\bf\hat{R}}\in{\cal\hat{R}}({\bf L}_{d,n^{*}})$ for some $n^{*}\in\mathbb{N}$.
\end{lemma}
\begin{proof}
	The proof is by contradiction.  We prove some initial claims, but first introduce a classical definition of the limit inferior of any sequence of sets $\{S_{n}:n\in\mathbb{N}\}$.
\begin{align}
	\label{liminf}
	\liminf_{n} S_{n}&:=\cup^{\infty}_{n=1}\cap^{\infty}_{q=n}S_{q}.
\end{align}
{\bf Claim \ref{grandefinale}.1} \textit{$\liminf\limits_{n} U_{n}\neq\{\}$.}
\begin{proof}
The proof is by contradiction.  Suppose $\liminf_{n}U_{n}=\{\}$ and Lemma \ref{grandefinale} is false.  If Lemma \ref{grandefinale} is false then $Z_{n}\neq{\cal S}(d)$ for all $n\in\mathbb{N}$.  For every $s\in{\cal S}(d)$, it follows from (\ref{liminf}) and the hypothesis $\liminf_{n}U_{n}=\{\}$ that some $n_{s}$ exists such that $s\notin U_{n_{s}}$.  Therefore (\ref{inductionupdate}) implies that $s\in Z_{n_{s}+1}$.  Moreover, (\ref{inductionupdate}) also implies $Z_{n}\subseteq Z_{n+1}$ for all $n\in\mathbb{N}$.  It follows that $Z_{n}={\cal S}(d)$, for all $n\geq\max_{s\in{\cal S}(d)}n_{s}$, which contradicts the hypothesis that $Z_{n}\neq{\cal S}(d)$ for all $n\in\mathbb{N}$.
\end{proof}
Define: 
\begin{align}
\label{U}
U^{*}&:=\liminf\limits_{n} U_{n}.
\end{align}

\hspace{-3mm}{\bf Claim \ref{grandefinale}.3} $\liminf\limits_{n}(U_{n}\setminus U^{*})=\{\}$.
\begin{proof}
The proof is by contradiction.  Suppose $\liminf_{n}(U_{n}\setminus U^{*})\neq\{\}$.  For some $s\in{\cal S}(d)$ and $n_{s}\in\mathbb{N}$, it follows that $s\in U_{n}\setminus U^{*}$ for all $n\geq n_{s}$.  Therefore $s\in\cap_{q\geq n_{s}} U_{q}$.  It follows from (\ref{liminf}) and (\ref{U}) that $s\in U^{*}$ which is a contradiction.
\end{proof}
\hspace{-3.5mm}{\bf Claim \ref{grandefinale}.6} \textit{For some function $l(n)$, $K_{1}\in\mathbb{N}$, and all $n\geq K_{1}$:}
\begin{align}
	\label{FDU}
	U^{*}&\subseteq\displaystyle\cup^{|{\bf L}_{d,n}|-1}_{q=l(n)}L_{q,n},\\
	\label{FDU2}
	{\cal S}(d)\setminus U^{*}&\subseteq\cup^{l(n)}_{q=0}L_{q,n}.
\end{align}
\begin{proof}
We introduce some preliminary notation.  For all $i,j\in {\cal S}(d)$ define:
\begin{align}
	\label{enij}
	e_{n}(i,j):=\text{\sc layer}_{n}(i)-\text{\sc layer}_{n}(j).
\end{align}
To prove (\ref{FDU}) and (\ref{FDU2}), it suffices to establish the following limits.  If $i,j\in U^{*}$ then:
\begin{align}
	\label{limenij1}
	\lim_{n\rightarrow\infty}e_{n}(i,j)&=e(i,j),
\end{align}
where $e(i,j)$ is a constant with respect to $n$.  If $i\in U^{*}$ and $j\in {\cal S}(d)\setminus U^{*}$ then:
\begin{align}
	\label{limenij2}
	\lim_{n\rightarrow\infty}e_{n}(i,j)&=\infty.
\end{align}
The definition of $U^{*}$ in (\ref{U}) implies that $U^{*}\subseteq U_{n}$ for all $n\geq K_{1}$ and some $K_{1}\in\mathbb{N}$.  For all $i\in U^{*}$ and $n\geq K_{1}$, it follows that:
\begin{align}
	\label{grandfinale81}
	\text{\sc layer}_{n+1}(i)&=\text{\sc layer}_{n}(i)+1,
\end{align}
where (\ref{grandfinale81}) follows from (\ref{shift4}) and the fact that $U^{*}\subseteq U_{n}$ for all $n\geq K_{1}$.  First, we justify (\ref{limenij1}).  Fix $i,j\in U^{*}$.  For all $n\geq K_{1}$:
\begin{align}
	\label{grandfinale84}
	e_{n+1}(i,j)&=\text{\sc layer}_{n+1}(i)-\text{\sc layer}_{n+1}(j),\\
	\label{grandfinale85}
	&=(\text{\sc layer}_{n}(i)+1)+(\text{\sc layer}_{n}(j)+1),\\
	\nonumber
	&=\text{\sc layer}_{n}(i)-\text{\sc layer}_{n}(j),\\
	\label{grandfinale87}
	&=e_{n}(i,j)
\end{align}
where (\ref{grandfinale84}) follows from (\ref{enij}), (\ref{grandfinale85}) follows from (\ref{grandfinale81}) and because $n\geq K_{1}$, and (\ref{grandfinale87}) follows from (\ref{enij}). Set $e(i,j):=e_{n}(i,j)$ for $n=K_{1}$.  Now (\ref{grandfinale87}) implies that $e_{n}(i,j)=e(i,j)$ for all $n\geq K_{1}$, which proves (\ref{limenij1}).

Now we justify (\ref{limenij2}).  Fix $i\in U^{*}$ and $j\in{\cal S}(d)\setminus U^{*}$.  Invoking Claim \ref{grandefinale}.3, we construct an infinite sequence $\{n_{q}:q\in\mathbb{N}\}$ where $n_{0}=K_{1}$ such that $j\notin U_{n_{q}}\setminus U^{*}$.  First we show that $j$ satisfies:
\begin{align}
	\label{grandfinale83}
	\text{\sc layer}_{n_{q}+1}(j)=\text{\sc layer}_{n_{q}}(j)\hspace{2mm}\text{for all}\hspace{1mm}q\in\mathbb{N}.
\end{align}
Since $j\in{\cal S}(d)\setminus U^{*}$, it follows that $j\notin U^{*}$.  Since $j\notin U_{n_{q}}\setminus U^{*}$ and $j\notin U^{*}$, it follows that $j\notin U_{n_{q}}$.  Therefore (\ref{grandfinale83}) follows from (\ref{shift4}).  We have the following:  
\begin{align}
	\label{grandfinale93}
	e_{n_{q}+1}(i,j)&=\text{\sc layer}_{n_{q}+1}(i)-\text{\sc layer}_{n_{q}+1}(j),\\
	\label{grandfinale94}
	&=(\text{\sc layer}_{n_{q}}(i)+1)-\text{\sc layer}_{n_{q}+1}(j),\\
	\label{pregrandfinale95}
	&=(\text{\sc layer}_{n_{q}}(i)+1)-\text{\sc layer}_{n_{q}}(j),\\
	\nonumber
	&=\text{\sc layer}_{n_{q}}(i)-\text{\sc layer}_{n_{q}}(j)+1,\\
	\label{grandfinale96}
	&=e_{n_{q}}(i,j)+1,
\end{align}
where (\ref{grandfinale93}) follows from (\ref{enij}), (\ref{grandfinale94}) follows from (\ref{grandfinale81}) and because $i\in U^{*}$ and $n_{0}=K_{1}$, (\ref{pregrandfinale95}) follows from (\ref{grandfinale83}), and (\ref{grandfinale96}) follows from (\ref{enij}).  More generally, for all $n\geq K_{1}$:
\begin{align}
	\label{grandfinale97}
	e_{n+1}(i,j)&=\text{\sc layer}_{n+1}(i)-\text{\sc layer}_{n+1}(j),\\
	\label{grandfinale98}
	&=(\text{\sc layer}_{n}(i)+1)-\text{\sc layer}_{n+1}(j),\\
	\label{grandfinale99}
	&\geq(\text{\sc layer}_{n}(i)+1)-(\text{\sc layer}_{n}(j)+1),\\
	\nonumber
	&=\text{\sc layer}_{n}(i)-\text{\sc layer}_{n}(j),\\
	\label{grandfinale100}
	&=e_{n}(i,j),
\end{align}
where (\ref{grandfinale97}) follows from (\ref{enij}), (\ref{grandfinale98}) follows from (\ref{grandfinale81}) and because $n\geq K_{1}$, (\ref{grandfinale99}) follows because (\ref{shift4}) implies $\text{\sc layer}_{n}(j)\leq\text{\sc layer}_{n+1}(j)\leq\text{\sc layer}_{n}(j)+1$, and (\ref{grandfinale100}) follows from (\ref{enij}).  Combining (\ref{grandfinale96}) and (\ref{grandfinale100}) implies $\lim_{n\rightarrow\infty}e_{n}(i,j)=\infty$ if $i\in U^{*}$ and $j\in {\cal S}(d)\setminus U^{*}$, which proves (\ref{limenij2}).
\end{proof}
We can now complete the proof of Lemma \ref{grandefinale}.  The proof is by contradiction.  If the Lemma is false, then Claim \ref{grandefinale}.1 implies there exists some non-empty $U^{*}$ defined by (\ref{U}) for all $n\geq K_{1}$ and some $K_{1}\in\mathbb{N}$.  For some function $l(n)$, all $q\in\{l(n)+1,\ldots,|{\bf L}_{d,n}|-1\}$ and all $n\geq K_{1}$, Claim \ref{grandefinale}.6 implies:   
\begin{align}
	\label{aln2}
	A_{q,n}(U^{*})&=L_{q,n}.
\end{align}
To justify (\ref{aln2}), suppose by contradiction, there is some $i\in L_{q,n}$ and $q\in\{l(n)+1,\ldots,|{\bf L}_{d,n}|-1\}$ such that $i\notin A_{q,n}(U^{*})$ for some $n\geq K_{1}$.  If $i\notin A_{q,n}(U^{*})$ then (\ref{aln}) implies $i\notin U^{*}$ because $A_{q,n}(U^{*}):=L_{q,n}\cap U^{*}$.  Since  $i\in{\cal S}(d)\setminus U^{*}$, it follows from(\ref{FDU2}) that $i\in\cup^{l(n)}_{q=0}L_{q,n}$ which contradicts the hypothesis.  
For all $q\in\{l(n)+1,\ldots,|{\bf L}_{d,n}|-1\}$ and $n\geq K_{1}$: 
\begin{align}
	\nonumber
	&\hspace{-3mm}\tilde{A}_{q,n}(U^{*})\\
	\label{newgrandefinale2}
	&\hspace{5mm}=(\cup^{q}_{k=0}L_{k,n})\setminus A_{q,n}(U^{*}),\\
	\label{newgrandefinale3}
	&\hspace{5mm}=(\cup^{l(n)}_{k=0}L_{k,n})\cup(\cup^{q}_{k=l(n)+1}L_{k,n})\setminus A_{q,n}(U^{*}),\\
	\label{newgrandefinale4}
	&\hspace{5mm}=({\cal S}(d)\setminus U^{*})\cup(\cup^{q}_{k=l(n)+1}L_{k,n})\setminus A_{q,n}(U^{*}),\\
	\label{newgrandefinale6}
	&\hspace{5mm}=({\cal S}(d)\setminus U^{*})\cup(\cup^{q-1}_{k=l(n)+1}A_{k,n}(U^{*})),
\end{align}
where (\ref{newgrandefinale2}) follows from (\ref{alntilde}), (\ref{newgrandefinale3}) follows because $q\in\{l(n)+1,\ldots,|{\bf L}_{d,n}|-1\}$ by assumption, (\ref{newgrandefinale4}) follows from (\ref{FDU2}), and (\ref{newgrandefinale6}) follows from (\ref{aln2}).  

For every $s\notin U^{*}$, Claim \ref{grandefinale}.3 implies $s\notin U_{n_{s}}\setminus U^{*}$ for some $n_{s}\in\mathbb{N}$.  Therefore $s\notin U_{n_{s}}$.  Lemma \ref{induction} implies $s\in Z_{n_{s}+1}$.  Define $K_{2}:=\max_{s\notin U^{*}}n_{s}+1$.  Since Lemma \ref{induction} also implies $Z_{n}\subseteq Z_{n+1}$ for all $n\in\mathbb{N}$, it follows that $({\cal S}(d)\setminus U^{*})\subseteq Z_{n}$ for all $n\geq K_{2}$.  For all $n\geq\max\{K_{1},K_{2}\}$:
\begin{align}
	\nonumber
	\hat{R}_{U_{n}\setminus U^{*}}&<\sum_{i\in U_{n}\setminus U^{*}}H(X_{i}\hat{Y}_{i})\\
	\label{newgrandfinale12}
	&\hspace{0.5mm}+\sum^{|{\bf L}_{d,n}|-1}_{q=0}h[A_{q,n}(U_{n}\setminus U^{*})|\tilde{A}_{q,n}(U_{n}\setminus U^{*})],\\
	\nonumber
	&=\sum_{i\in U_{n}\setminus U^{*}}H(X_{i}\hat{Y}_{i})\\
	\label{newgrandfinale13}
	&\hspace{3.5mm}+\sum^{l(n)}_{q=0}h[A_{q,n}(U_{n}\setminus U^{*})|\tilde{A}_{q,n}(U_{n}\setminus U^{*})],
\end{align}
where (\ref{newgrandfinale12}) follows because $(U_{n}\setminus U^{*})\subseteq({\cal S}(d)\setminus U^{*})\subseteq Z_{n}$ and all subsets of $Z_{n}$ satisfy (\ref{compdecomp}) by definition, and (\ref{newgrandfinale13}) follows from (\ref{FDU2}).  For all $n\geq\max\{K_{1},K_{2}\}$:
\begin{align}
	\nonumber
	\hat{R}_{U_{n}}&>\sum_{i\in U_{n}}H(X_{i}\hat{Y}_{i})\\
	\label{newgrandfinale7}
	&\hspace{4mm}+\sum^{|{\bf L}_{d,n}|-1}_{q=0}h[A_{q,n}(U_{n})|\tilde{A}_{q,n}(U_{n})],\\
	\nonumber
	&>\sum_{i\in U_{n}\setminus U^{*}}H(X_{i}\hat{Y}_{i})+\sum_{i\in U^{*}}H(X_{i}\hat{Y}_{i})\\
	\nonumber
	&\hspace{8mm}+\sum^{l(n)}_{q=0}h[A_{q,n}(U_{n})|\tilde{A}_{q,n}(U_{n})]\\
	\label{newgrandfinale8}
	&\hspace{12mm}+\sum^{|{\bf L}_{d,n}|-1}_{q=l(n)+1}h[A_{q,n}(U_{n})|\tilde{A}_{q,n}(U_{n})],\\
	\nonumber
	&>\sum_{i\in U_{n}\setminus U^{*}}H(X_{i}\hat{Y}_{i})+\sum_{i\in U^{*}}H(X_{i}\hat{Y}_{i})\\
	\nonumber
	&\hspace{8mm}+\sum^{l(n)}_{q=0}h[A_{q,n}(U_{n})|\tilde{A}_{q,n}(U_{n})]\\
	\label{newgrandfinale9}
	&\hspace{12mm}+\sum^{|{\bf L}_{d,n}|-1}_{q=l(n)+1}h[A_{q,n}(U^{*})|\tilde{A}_{q,n}(U^{*})],\\
	\nonumber
	&>\sum_{i\in U_{n}\setminus U^{*}}H(X_{i}\hat{Y}_{i})+\sum_{i\in U^{*}}H(X_{i}\hat{Y}_{i})\\
	\nonumber
	&\hspace{8mm}+\sum^{l(n)}_{q=0}h[A_{q,n}(U_{n}\setminus U^{*})|\tilde{A}_{q,n}(U_{n}\setminus U^{*})]\\
	\label{newgrandfinale9a}
	&\hspace{12mm}+\sum^{|{\bf L}_{d,n}|-1}_{q=l(n)+1}h[A_{q,n}(U^{*})|\tilde{A}_{q,n}(U^{*})],\\
	\nonumber
	&>\sum_{i\in U_{n}\setminus U^{*}}H(X_{i}\hat{Y}_{i})\\
	\nonumber
	&\hspace{4mm}+\sum^{l(n)}_{q=0}h[A_{q,n}(U_{n}\setminus U^{*})|\tilde{A}_{q,n}(U_{n}\setminus U^{*})]\\
	\label{newgrandfinale10}
	&\hspace{8mm}+\sum_{i\in U^{*}}H(X_{i}\hat{Y}_{i})+h[U^{*}|({\cal S}(d)\setminus U^{*})],\\
	\nonumber
	&=\sum_{i\in U^{*}}H(X_{i}\hat{Y}_{i})+H(X_{U^{*}}\hat{Y}_{U^{*}}|X_{{\cal S}(d)\setminus U^{*}}\hat{Y}_{{\cal S}(d)\setminus U^{*}}Y_{d})\\
	\nonumber
	&\hspace{4mm}+\sum_{i\in U_{n}\setminus U^{*}}H(X_{i}\hat{Y}_{i})\\
	\label{newgrandfinale11}
	&\hspace{8mm}+\sum^{l(n)}_{q=0}h[A_{q,n}(U_{n}\setminus U^{*})|\tilde{A}_{q,n}(U_{n}\setminus U^{*})],
\end{align}
where (\ref{newgrandfinale7}) follows from the definition of $U_{n}$, (\ref{newgrandfinale8}) follows by splitting the sums in (\ref{newgrandfinale7}), (\ref{newgrandfinale9}) follows from (\ref{newgrandfinale8}) and (\ref{FDU}), (\ref{newgrandfinale9a}) follows from (\ref{newgrandfinale9}) and (\ref{FDU2}), (\ref{newgrandfinale10}) follows from the definition of $h[\hspace{0.5mm}\cdot\hspace{0.5mm}|\hspace{0.5mm}\cdot\hspace{0.5mm}]$ in (\ref{hdef}), (\ref{FDU2}), (\ref{newgrandefinale6}), and the chain rule, and (\ref{newgrandfinale11}) follows from the definition of $h[\hspace{0.5mm}\cdot\hspace{0.5mm}|\hspace{0.5mm}\cdot\hspace{0.5mm}]$ in (\ref{hdef}).

Now $\hat{R}_{U_{n}}=\hat{R}_{U_{n}\setminus U^{*}}+\hat{R}_{U^{*}}$.  Since $n\geq\max\{K_{1},K_{2}\}$ we can invoke (\ref{newgrandfinale13}).  Together (\ref{newgrandfinale13}) and (\ref{newgrandfinale11}) imply:
\begin{align}
	\nonumber
	\hat{R}_{U^{*}}&>\sum_{i\in U^{*}}H(X_{i}\hat{Y}_{i})\\
	\label{actualgrandefinale}
	&\hspace{10mm}+H(X_{U^{*}}\hat{Y}_{U^{*}}|X_{{\cal S}(d)\setminus U^{*}}\hat{Y}_{{\cal S}(d)\setminus U^{*}}Y_{d}).
\end{align}
By selection ${\bf\hat{R}}\in{\cal\hat{R}}_{d}$ satisfies (\ref{boundary}) for all $S\subseteq{\cal S}(d)$, which contradicts (\ref{actualgrandefinale}).  It follows that the non-empty $U^{*}$ defined in (\ref{U}) does not exist.  Therefore Claim \ref{grandefinale}.1 implies ${\bf\hat{R}}\in{\cal\hat{R}}({\bf L}_{d,n})$ for some $n\in\mathbb{N}$.  
\end{proof}